\documentclass[runningheads]{llncs}
\usepackage{graphicx}
\usepackage{amsmath}
\usepackage{amsthm}
\usepackage{amssymb}
\usepackage{ wasysym }
\usepackage[colorlinks=true, allcolors=blue]{hyperref}

\usepackage{siunitx}
\usepackage{listings}
\usepackage[dvipsnames]{xcolor}
\usepackage{graphicx}
\usepackage{stmaryrd}
\usepackage{orcidlink}
\usepackage{mathtools}
\usepackage[frozencache=true,cachedir=minted-cache]{minted} 
\usepackage[colorinlistoftodos]{todonotes}
\usepackage{tikz, tikz-cd}
\usetikzlibrary{automata, arrows.meta, positioning, calc, decorations.pathreplacing, backgrounds, matrix, arrows, patterns}

\usetikzlibrary{decorations,decorations.markings} 
\pgfkeys{/tikz/.cd,
    contour distance/.store in=\ContourDistance,
    contour distance=-10pt, 
    contour step/.store in=\ContourStep,
    contour step=1pt,
}

\pgfdeclaredecoration{closed contour}{initial}
{%
\state{initial}[width=\ContourStep,next state=cont] {
    \pgfmoveto{\pgfpoint{\ContourStep}{\ContourDistance}}
    \pgfcoordinate{first}{\pgfpoint{\ContourStep}{\ContourDistance}}
    \pgfpathlineto{\pgfpoint{0.3\pgflinewidth}{\ContourDistance}}
    \pgfcoordinate{lastup}{\pgfpoint{1pt}{\ContourDistance}}
    
  }
  \state{cont}[width=\ContourStep]{
     \pgfmoveto{\pgfpointanchor{lastup}{center}}
     \pgfpathlineto{\pgfpoint{\ContourStep}{\ContourDistance}}
     \pgfcoordinate{lastup}{\pgfpoint{\ContourStep}{\ContourDistance}}
  }
  \state{final}[width=\ContourStep]
  { 
    \pgfmoveto{\pgfpointanchor{lastup}{center}}
    \pgfpathlineto{\pgfpointanchor{first}{center}}
  }
}

\makeatletter
\newsavebox{\@brx}
\newcommand{\llangle}[1][]{\savebox{\@brx}{\(\m@th{#1\langle}\)}%
  \mathopen{\copy\@brx\mkern2mu\kern-0.9\wd\@brx\usebox{\@brx}}}
\newcommand{\rrangle}[1][]{\savebox{\@brx}{\(\m@th{#1\rangle}\)}%
  \mathclose{\copy\@brx\mkern2mu\kern-0.9\wd\@brx\usebox{\@brx}}}

  \newcommand{\nxt}{\bigcirc}
    \newcommand{\until}{\, \Large U \,}

  \newcommand{\ltl}[0]{\text{LTL}}
  \newcommand{\ltln}[0]{\text{LTL}_{\setminus \nxt }}
  \newcommand{\ctln}[0]{\text{CTL}^*_{\setminus \nxt }}

  \definecolor{lpurple}{RGB}{251,181,255}
  \definecolor{lred}{RGB}{255,170,170}
  \definecolor{lcyan}{RGB}{223,255,252}
  \definecolor{black_opac}{RGB}{0,0,0}
  \tikzset{%
  prefix node name/.code={%
    \tikzset{%
      name/.code={\edef\tikz@fig@name{#1 ##1}}
    }%
  }%
}

\definecolor{color1}{rgb}{0.1,0.498039215686275,0.9549019607843137}
\definecolor{alizarin}{rgb}{0.82, 0.1, 0.26}
\definecolor{antiquewhite}{rgb}{0.98, 0.92, 0.84}
\definecolor{azure}{rgb}{0.94, 1.0, 1.0}
\definecolor{offwhite}{rgb}{0.98, 0.97, 0.97}
\definecolor{pigment}{rgb}{0.2, 0.2, 0.6}

\makeatother

\begin{document}

\title{Bisimulation Learning}
\author{Alessandro Abate\inst{1}
 \and
 Mirco Giacobbe \inst{2}
\and
 Yannik Schnitzer\inst{1}
 }
 \institute{University of Oxford, UK\\
 \email{\{alessandro.abate,yannik.schnitzer\}@cs.ox.ac.uk}\\
 \and
    University of Birmingham, UK\\
    \email{m.giacobbe@bham.ac.uk}}
\maketitle              

\begin{abstract}
We introduce a data-driven approach to computing finite bisimulations for state transition systems with very large, possibly infinite state space. Our novel technique computes stutter-insensitive bisimulations of deterministic systems, which we characterize as the problem of learning a state classifier together with a ranking function for each class. Our procedure learns a candidate state classifier and candidate ranking functions from a finite dataset of sample states; then, it checks whether these generalise to the entire state space using satisfiability modulo theory solving. Upon the affirmative answer, the procedure concludes that the classifier constitutes a valid stutter-insensitive bisimulation of the system. Upon a negative answer, the solver produces a counterexample state for which the classifier violates the claim, adds it to the dataset, and repeats learning and checking in a counterexample-guided inductive synthesis loop until a valid bisimulation is found. 
We demonstrate on a range of benchmarks from reactive verification and software model checking that our method yields faster verification results than alternative state-of-the-art tools in practice. Our method produces succinct abstractions that enable an effective verification of linear temporal logic without next operator, and are interpretable for system diagnostics.

\end{abstract}
\section{Introduction}
Abstraction of state transition systems is the process for which 
a system under analysis---the concrete system---is 
reduced to another system---the abstract system---that is simpler to analyze and preserves certain temporal properties of the former~\cite{DBLP:conf/popl/CousotC77,DBLP:conf/concur/Glabbeek93,DBLP:books/sp/Milner80,DBLP:conf/tcs/Park81}. 
It is a fundamental approach to state space reduction in the verification of finite-state systems and an essential element for the verification of infinite-state systems. 
Bisimulations are the abstractions that preserve linear 
and branching behaviour with respect to propositional observations, for which the model checking question 
for both linear- and branching-time logics have the same answer on the abstract and the concrete system~\cite{DBLP:journals/tcs/BrowneCG88,DBLP:journals/jacm/HennessyM85}.

Computing a bisimulation amounts to
computing an equivalence relation 
on the state space that 
is stable with respect to a notion of state change, 
and preserves propositional observations. 
An equivalence relation defines a partition 
of the concrete state space 
and induces an abstract system where every abstract state 
corresponds to an equivalence class.
The problem of computing bisimulations over an explicit representation of the state graph has been widely 
studied in the past~\cite{DBLP:journals/fac/BalcazarGS92,DBLP:journals/iandc/KanellakisS90}, since Hopcroft's graph minimisation algorithm and the
Paige-Tarjan algorithm for iterative partition refinement~\cite{HOPCROFT1971189,DBLP:journals/siamcomp/PaigeT87}. 
 Partition refinement was improved with on-the-fly partition refinement of the reachable state space 
as well as parallelisation~\cite{DBLP:journals/sttt/DijkP18,DBLP:conf/stoc/LeeY92,DBLP:conf/cav/LeeR94,DBLP:conf/facs2/0001GHHW21}. 
Yet, explicit-state algorithms fall short on systems with 
very large or infinite state space, for which one must resort to procedures that represent regions of state space 
symbolically~\cite{DBLP:conf/cav/BouajjaniFH90}. 

Partition refinement relies on computing exact pre- and post-images through the transition function of the system~\cite{DBLP:conf/cav/BouajjaniFH90,DBLP:journals/tocl/GrooteJKW17}.
This entails quantifier elimination, 
which is computationally costly. 
Counterexample-guided abstraction refinement (CEGAR) provides an approach to avoid pre- and post-image computation; it computes {\em simulations} of state transition systems incrementally, from infeasibility proofs of spurious counterexamples~\cite{DBLP:conf/cav/ClarkeGJLV00,DBLP:conf/popl/HenzingerJMM04}. 
The resulting abstract system is tight enough to verify a specific property of interest, but cannot generally provide concrete counterexamples when a property is false and, for this purpose, methods based on CEGAR are usually coupled with bounded model checking~\cite{DBLP:journals/ac/BiereCCSZ03}. Similarly, methods for temporal logic verification based on proof rules (i.e., certificates) provide sufficient conditions to verify whether a property holds but do not provide a counterexample when this is false~\cite{DBLP:journals/iandc/GrumbergFR85,DBLP:journals/apal/Vardi91,DBLP:conf/popl/CookGPRV07,DBLP:conf/hybrid/MuraliTZ24,cav24supermartingales}.
By contrast, bisimulations provide a tight abstraction where abstract counterexamples correspond to concrete counterexamples and, as such, these are directly interpretable for system debugging and diagnostics.

We present a data-driven approach to computing finite bisimulations 
from sample states and transitions of the system, which skips partition refinement entirely.
We adapt the notion of \textit{well-founded bisimulations}, where the condition of stability of the equivalence relation with respect to stuttering is characterised as the existence of ranking functions over well-founded sets~\cite{DBLP:conf/fsttcs/Namjoshi97}. 
While originally introduced solely as a proof rule, we leverage well-founded bisimulations for the first time to directly compute finite bisimulations.  
We instantiate well-founded bisimulations with ranking functions that, for every state transition to a different state in the abstract system, map states to natural numbers that decrease strictly as the system stutters. 
This characterises {\em stutter-insensitive bisimulations} for deterministic transition systems and also applies to strong bisimulations, which is the special case of our method where ranking functions are constant.

Stutter-insensitive bisimulations are stable bisimulations with respect to observation change in the system, and is closed with respect to all state transitions between these changes. A system stutters when it 
changes concrete state without changing observation~\cite{DBLP:conf/ifip/Lamport83}, 
and stutter-insensitive bisimulations abstract stuttering away. In contrast to strong bisimulations, stutter-insensitive bisimulations result in much more succinct abstractions, while being sufficiently strong to preserve the validity of any linear temporal logic specification without next operator. 
While our approach also applies to strong bisimulations, we generalise our method to stutter-insensitive bisimulations, 
because they more effectively yield finite abstractions 
on infinite-state systems in practice.

We build on the observation that a finite partition can be characterized as a state classifier mapping the (possibly infinite) state space into a finite set of classes. This reduces the problem of computing a stutter-insensitive bisimulation to training a classifier and a ranking function for each class~\cite{DBLP:conf/sas/UrbanM14,DBLP:conf/sigsoft/Nori013,DBLP:conf/sigsoft/GiacobbeKP22}. For the partition classifier, we employ binary a decision tree (BDT) with parametric linear predicates at each decision node, and we associate each leaf node with a parametric linear ranking function. This structure forms our template.

Our approach is underpinned by a learner and a verifier interacting with each other, both using a satisfiability modulo theory (SMT) solver. The learner proposes a candidate bisimulation by computing parameters of the classifier and ranking function templates to satisfy conditions over sampled transitions. The verifier then checks if these conditions hold over the \textit{entire state space}. If affirmed, the classifier induces a stutter-insensitive bisimulation. If not, the verifier provides a counterexample, a state where stutter-insensitive bisimulation conditions are violated. This counterexample is fed back to the learner, which updates the classifier and ranking functions. The process repeats in a counterexample-guided inductive synthesis (CEGIS) loop until the verifier confirms the bisimulation's validity~\cite{DBLP:conf/asplos/Solar-LezamaTBSS06}. If the template cannot fit 
the finite set of samples, for instance, due to an insufficient number of partitions, our procedure automatically enlarges the BDT with an additional layer and resumes the CEGIS loop.

We demonstrate the experimental efficacy of our approach on numerical programs and reactive software systems with integer state spaces. We consider benchmarks from reactive verification 
and software model checking, in particular discrete-time synchronisation protocols and conditional termination analysis problems. We benchmark the former set against the nuXmv model checker for reactive verification and the latter against the Ultimate and the CPAChecker tools for software verification~\cite{DBLP:conf/cav/BeyerK11,DBLP:conf/cav/CavadaCDGMMMRT14,DBLP:conf/tacas/HeizmannBDFHKNSSP23}. The results are two-fold. For the reactive verification benchmarks, our approach has faster verification times than nuXmv on systems with long stuttering intervals. For the conditional termination benchmarks, our approach is able to generate exact preconditions for which the program terminates, unlike the baselines that return negative answers when the program does not terminate for at least one input. In summary, we demonstrate that, on these problems, our approach yields both faster and more informative results than the alternative state-of-the-art tools.

We summarise our contributions in the following three points: 
    (1) we introduce the first data-driven approach to construct bisimulations, as an alternative approach to partition refinement;
    (2) we implement the theory of well-founded bisimulations which we synthesise in a CEGIS loop, as a means to compute stutter-insensitive bisimulations; 
    (3) we demonstrate the efficacy of our novel approach on reactive verification and software model checking benchmarks. 
 Our approach is fully automatic and requires no user input beyond the system itself. It produces succinct abstractions of infinite-state systems, which effectively enables their LTL (without next) verification using finite-state model checkers.

\section{Illustrative Example}

We motivate our procedure with an example from software model checking. Consider the code snippet in Figure~\ref{fig:euclid}a. The program takes two arbitrary integers as input and subtracts the smaller from the larger until the two values coincide. 
We ask the question of whether the program terminates for every initial condition, which is not straightforward to answer for this example. Given two positive inputs, the program run the Euclidean algorithm for the greatest common divisor and terminates once it is found. However, for any two unequal non-positive inputs, this implementation will never exit the loop and run forever. 

\vspace{0.2cm}
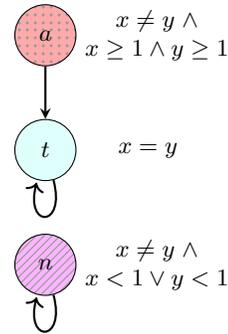
\begin{figure}
    \centering  
    \begin{tabular}{ccc}
        \begin{minipage}[b]{0.49\columnwidth}
        \centering
        \begin{minted}[mathescape]{java}
          int x = *, y = *;
          while (x != y) {
              if (x > y) 
                  x = x - y;
              else 
                  y = y - x;
          }


        \end{minted}
        \vfill
        \end{minipage}
        &
        \qquad
        &
        \begin{minipage}[b]{0.49\columnwidth}
        \centering
        \begin{tikzpicture}
          \begin{scope}[xshift = 4.5cm, prefix node name =G2, yshift = 2cm, node distance = .7cm]
        
        \node (q0) [state, initial text = {},preaction={fill, lred}, pattern = {dots},pattern color = black_opac!40!white]  {$a$};
        \node (l0) [right = 1cm and 0cm of G2 q0,align=center] {\footnotesize$x \not = y \; \wedge$\\ \footnotesize$x \geq 1 \land y \geq 1$};
        
        \node (q1) [state, below = of G2 q0,,fill = lcyan] {$t$};
        \node (l1) [right = 2cm and 0cm of G2 q1, align = center, inner sep = .55cm] {\footnotesize$x = y$};
        
        \node (q2) [state, below = of G2 q1,,preaction={fill, lpurple}, pattern = {north east lines}, pattern color = black_opac!40!white] {$n$};
        \node (l2) [right = 1cm and 0cm of G2 q2,align=center] {\footnotesize$x \not = y \; \wedge$\\ \footnotesize$x < 1 \lor y < 1$};
        
        \path [-stealth, thick]
            (G2 q0) edge node {$ $}   (G2 q1)
            (G2 q1) edge [loop below]  node {$ $}()
            (G2 q2) edge [loop below]  node {$ $}();
        \end{scope}
    \end{tikzpicture}
    \end{minipage}\\
    (a) Concrete program && 
    (b) Abstract program
    \end{tabular}
    \caption{Learned stutter-insensitive bisimulation of the Euclidean algorithm.}
    \label{fig:euclid}
\end{figure}
\vspace{0.2cm}

Our procedure solves the termination problem by iteratively learning parameters for a given state classifier template, such that its induced partition of the state space satisfies the stutter-insensitive bisimulation conditions over a finite set of sample transitions of the program. We ensure this by simultaneously computing parameters for given ranking function templates, which, together with the partition induced by the classifier, satisfy the equivalent conditions of a well-founded bisimulation. We leverage an SMT solver to check for counterexamples, i.e., states that are not equivalent to other states with the same class assigned by the classifier. These counterexample states are passed back to the learning procedure to update the classifier and the ranking function parameters until the SMT solver cannot generate a counterexample anymore and, thus, certifies that the learned classifier generalises to the entire infinite state space and induces a valid stutter-insensitive bisimulation.

Figure~\ref{fig:exProgress} illustrates the iterative update of the classifier with respect to the sampled program behaviour, given an initial partitioning of the state space into the class of \textit{terminated} states violating the loop condition \mintinline{java}{x != y} and the disjoint class of \textit{not terminated} states. Upon termination 
the learned classifier correctly separates the states into those for which both variables are positive and which will eventually reach a terminated state after stuttering for a finite number of steps and the states that infinitely stutter in the class of not-terminated states. 

\vspace{0.2cm}
\begin{figure}
    
     \centering

    \begin{tabular}{ccc}
\includegraphics[width=0.325\textwidth]{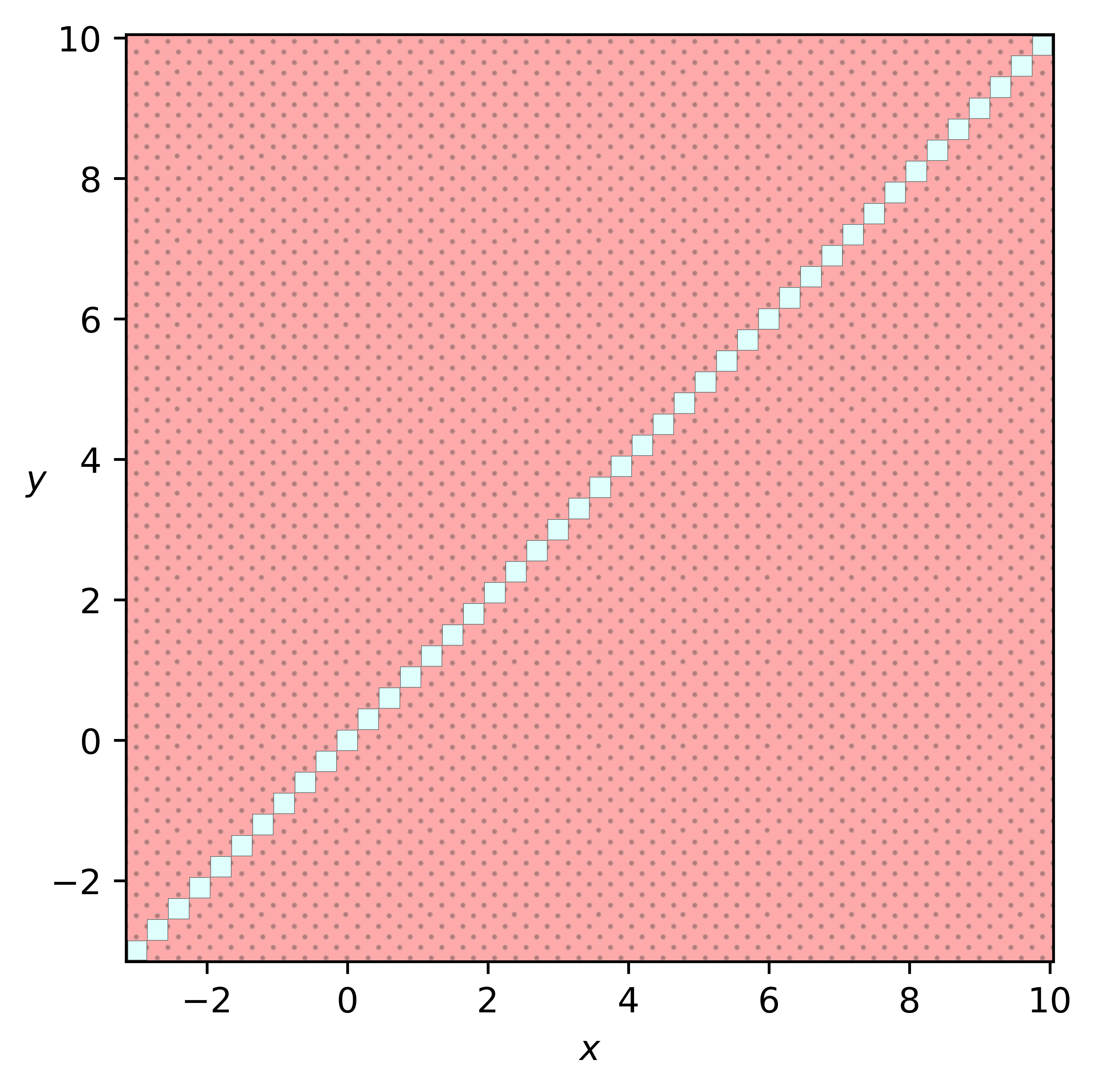}\hfill
&
\includegraphics[width=0.325\textwidth]{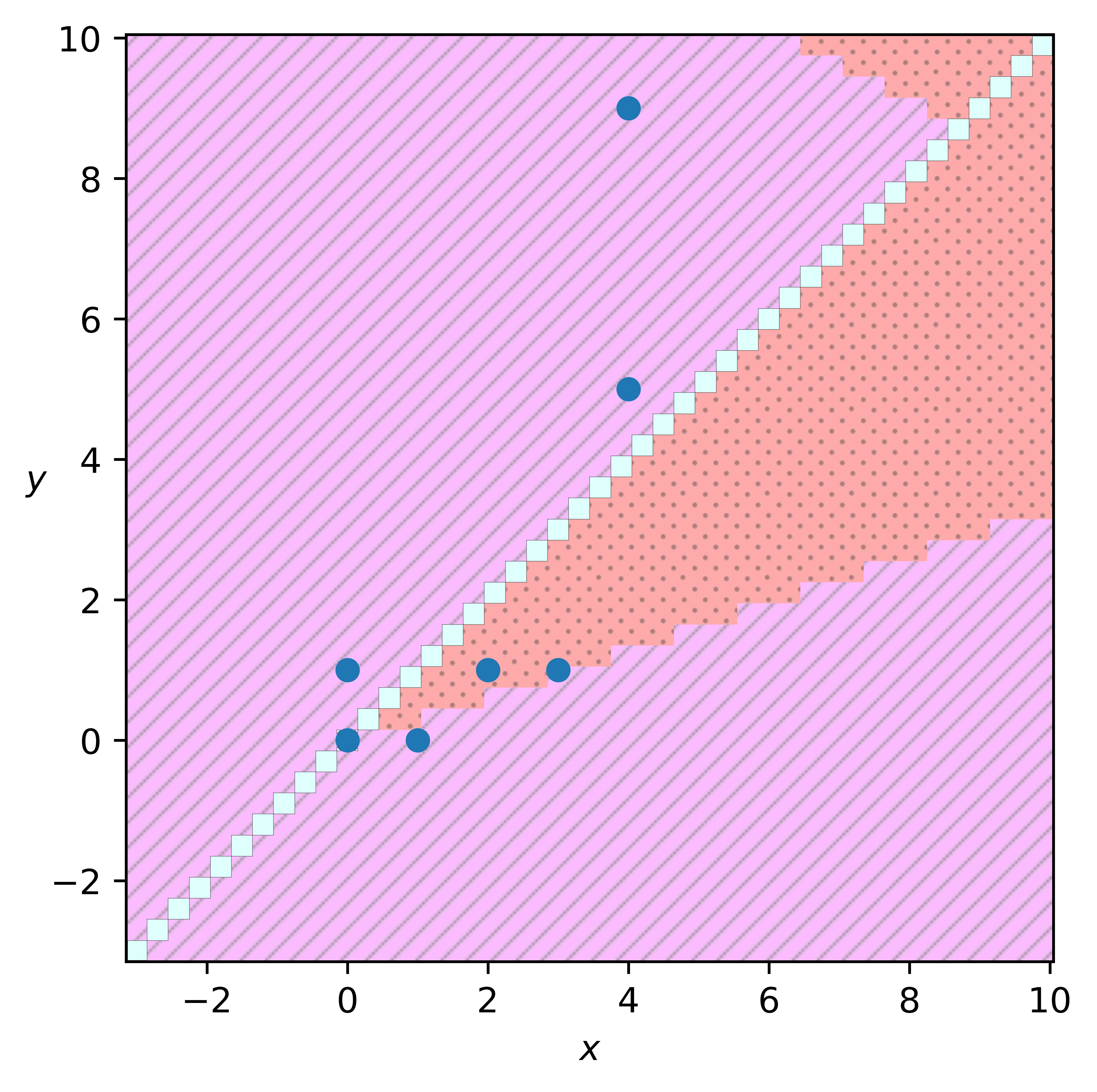}\hfill
&
\includegraphics[width=0.325\textwidth]{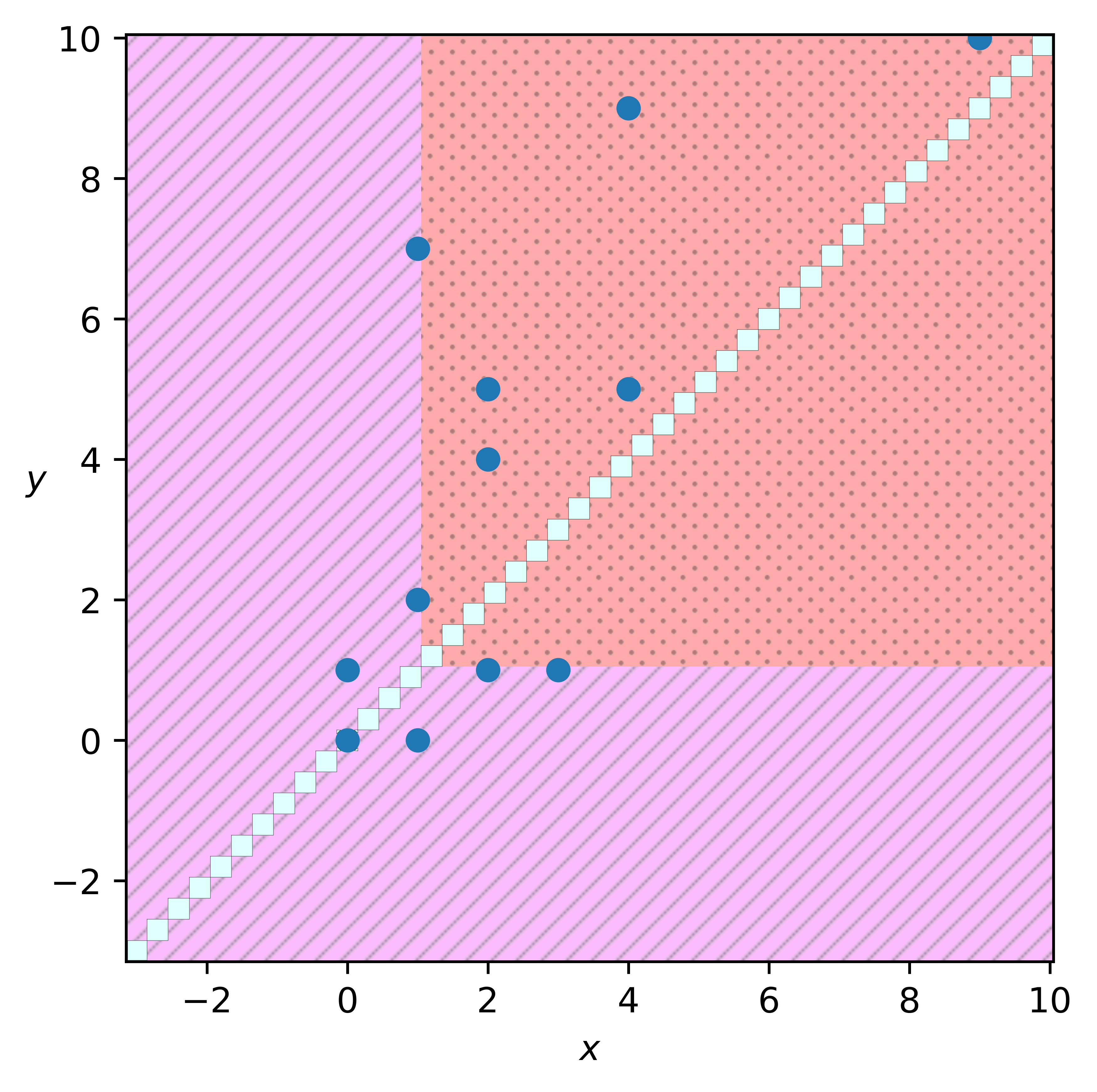}\\
(a) Initial partition &
(b) Intermediate partition & 
(c) Final partition
    \end{tabular}

        \caption{Iterative process of bisimulation learning. Starting from the initial label-preserving partition (a), our procedure generates counterexamples (blue dots) until it attains a valid stutter-insensitive bisimulation (c).}
        \label{fig:exProgress}
\end{figure}
\vspace{0.2cm}

In addition to the stutter-insensitive bisimulation, our procedure generates the corresponding abstract system by computing the behavior of the abstract states (i.e., the classes of the partition) alongside the classifier. Figure~\ref{fig:euclid}b shows the synthesized abstract system for the Euclidean algorithm, where each abstract state corresponds to an infinite subset of the concrete state space. The stutter-insensitive bisimulation ensures that the termination question has the same answer for all concrete states within the same class. A key advantage of our approach over methods providing a single counterexample is that it produces interpretable representations of the abstract system, aiding in system diagnostics. Specifically, our approach yields interpretable classifiers as binary decision trees. Figure~\ref{fig:euclid}b shows the abstract system and the automatically generated predicates defining the partition. Even for high-dimensional state spaces and complex partitions, this approach provides accessible means to interpret and diagnose the system for potential faults and undesired behavior~\cite{DBLP:conf/hybrid/AshokJJKWZ20,DBLP:conf/tacas/BrazdilCKT18}.

\section{Stutter-insensitive Bisimulations of Deterministic  Transition Systems}
We introduce the fundamental concepts underpinning our approach.

\begin{definition}[Transition Systems]
    A transition system ${\cal M}$ consists of
    \begin{itemize}
	   \item a state space $S$,
 	 \item an initial region $I \subseteq S$, and
	   \item a non-blocking transition function $T \colon S \to (2^S \setminus \emptyset)$.
\end{itemize}
We say that $\cal M$ is deterministic if $|T(s)| = 1$ for all $s \in S$. It is labelled when it additionally comprises
\begin{itemize}
    \item a set of atomic propositions $AP$ (the observables), and
    \item a labelling (or observation) function $\llangle \cdot \rrangle \colon S \to 2^{AP}$.
\end{itemize}
A trajectory of $\cal M$ is any sequence of states $\tau = s_0, s_1, s_2, \dots$ such that $s_{i+1} \in T(s_i)$ for all consecutive $s_i, s_{i+1}$ in $\tau$. We say that $\tau$ is initialised if $s_0 \in I$.
    \label{def:ts}
\end{definition}

\begin{definition}[Partitions]
     A partition on $\cal M$ is an equivalence relation $\simeq \subseteq S \times S$ on $S$,
     which defines the quotient space $S/_{\simeq}$ (i.e., the set of equivalence classes of $\simeq$) of pairwise-disjoint regions of $S$ whose union is $S$.
\end{definition}

Since we are interested in a notion of state equivalence insensitive to behaviour that does not change the observation of a state, the concept of divergence will be essential to distinguish between states that progress while not changing observation and those that do not progress at all~\cite{DBLP:books/daglib/baierkatoen,DBLP:conf/lics/Walker88}. 

\begin{definition}[Divergence Sensitivity]
    Let $\simeq$ be a partition on $\cal M$. A state $s \in S$ is $\simeq$-divergent if there exists an infinite trajectory $s_0, s_1,\dots$ such that $s_0 = s$ and $s_i \simeq s$ for all $i > 0$. Partition $\simeq$ is divergence-sensitive when $s \simeq t$ and $s$ is $\simeq$-divergent implies that $t$ is $\simeq$-divergent.
\end{definition}

A partition of the state space induces a reduced transition system --- the corresponding abstract system or \textit{quotient}.

\begin{definition}[Quotient]
     The quotient of $\cal M$ under the partition $\simeq$ is the transition system ${\cal M}/_\simeq$ with
\begin{itemize}  
    \item state space $S/_{\simeq}$,
    \item initial region $I/_\simeq$ where $R \in I/_\simeq$ iff $R \cap I \neq \emptyset$, and
\item transition function $T/_\simeq$ where 
\begin{enumerate}
    \item $R\neq Q \in T/_\simeq(R)$ iff $T(s) \in Q$ for some $s \in R$,
    \item $R \in T/_\simeq(R)$ iff some $s \in R$ is $\simeq$-divergent.
\end{enumerate}
\end{itemize}
\label{def:quotient}
\end{definition}

The quotient is the aggregation of equivalent states and their behaviours. The specifications preserved by the quotient, i.e., the statements that carry over from the abstract to the concrete system, depend on the properties of the underlying partition~\cite{DBLP:journals/siamcomp/PaigeT87}. The most important property to preserve sensible specifications is that equivalent states must have equal observations.

\begin{definition}[Label-preserving Partitions]
    A partition $\simeq$ on a labelled transition system is label-preserving when $s \simeq t$ implies $\llangle s \rrangle = \llangle t \rrangle$. The quotient ${\cal M}/_\simeq$ of a labelled transition system $\cal M$ under a label-preserving partition $\simeq$ is labelled with the extended labelling function $\llangle \cdot \rrangle \colon S \cup S/_\simeq \to 2^{AP}$ 
    where, for every region $R \in S/_\simeq$,  $\llangle R \rrangle = \llangle s \rrangle$ for any representative $s \in R$.
\end{definition}

A standard notion of state equivalence on labelled transition systems is \textit{bisimilarity}~\cite{DBLP:books/daglib/Milner89}. Bisimilarity preserves both linear- and branching-time behaviour by co-inductively requiring that every pair of related states can match each others' transitions with equivalent transitions. However, this stability with respect to stepwise behaviour often results in large quotients, thus limiting its suitability to facilitate reasoning over the system~\cite{DBLP:journals/siamcomp/PaigeT87}. Therefore, we focus on \textit{stutter-insensitive bisimulations}~\cite{DBLP:journals/tcs/BrowneCG88}. By abstracting from stepwise behaviour that does not change the observation of a state, stutter-insensitive bisimulations yield smaller quotients while preserving important specifications, as we will see in the following section.

\begin{definition}[Stutter-insensitive Bisimulation]
A label-preserving partition $\simeq$ is a stutter-insensitive bisimulation if, for all states $s, s', t \in S$ such that $s \simeq t$ and $s \not\simeq s' \in T(s)$, there exists a finite trajectory 
$t_0, t_1, \dots, t_k$ such that 
$t_0 = t$, $t_i \simeq s$ for all $i = 1, \dots k-1$, and $t_k = t'$ for some $t' \simeq s'$.
\label{def:stuttBisim}
\end{definition}

Figure~\ref{fig:stuttintuition} illustrates the stability condition of stutter-insensitive bisimulations. This condition requires that for related states, transitions to unrelated states can be matched by finite trajectories that pass through the same equivalence class.

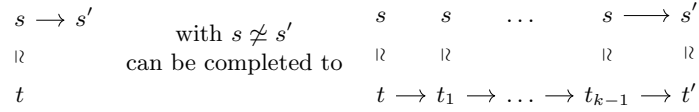
\begin{figure}
    \centering
        \begin{tikzcd}[row sep=normal, column sep = small]
     s               \rar{} \arrow[draw=none]{d}[sloped,auto=false]{\simeq}
       & s' \\
     t   
       &  
    \end{tikzcd}
    \parbox{10em}{\centering with $s \not \simeq s'$\\can be completed to}
    \begin{tikzcd}[row sep=normal, column sep = small]
     s                \arrow[draw=none]{d}[sloped,auto=false]{\simeq} & s \arrow[draw=none]{d}[sloped,auto=false]{\simeq} & \dots & s \arrow[draw=none]{d}[sloped,auto=false]{\simeq} \rar{}& s' \arrow[draw=none]{d}[sloped,auto=false]{\simeq}\\
     t  \rar{} & t_1 \rar{} & \dots  \rar{} & t_{k-1} \rar{} & t'
    \end{tikzcd}
    \caption{Trajectory-based representation of the stutter-insensitive stability condition.}
    \label{fig:stuttintuition}
    \vspace{-3mm}
\end{figure}

\begin{lemma}
Every stutter-insensitive bisimulation on any deterministic labelled transition system admits a deterministic quotient. 
\label{lem:detQuo}
\end{lemma}
\begin{proof}
    Let $\cal M$ be a deterministic transition system and $\simeq$ be a stutter-insensitive bisimulation on $\cal M$. Assume $M/_{\simeq}$ is nondeterministic, this implies that there exists pairwise distinct $R, Q, V \in S/_{\simeq}$ such that $\{Q,V\} \subseteq T/_{\simeq}(R)$. It follows that there exist $s,t \in R$ with $T(s) \in Q$ and $T(t) \in V$. Since $s,t \in R$ it holds that $s \simeq t$ and as $\cal M$ is deterministic and $Q \not = V$, $\simeq$ cannot satisfy Def.~\ref{def:stuttBisim}. \qed
\end{proof}

\subsection{Model Checking}

We introduce Linear Temporal Logic without \textit{next}-operator ($\ltln$) as a formal specification language for the temporal behaviour of a system and its states~\cite{DBLP:books/daglib/baierkatoen,DBLP:conf/focs/Pnueli77}. $\ltln$ formulas are constructed according to the following grammar:
$$
    \varphi ::= \text{true} \mid p \mid \varphi \wedge \varphi \mid \neg \varphi \mid \varphi \until \varphi
$$
The model checking problem for $\ltln$ is to decide whether transition system $\cal M$ satisfies a given $\ltln$ formula $\varphi$, where the satisfaction relation $\models$ for trajectories of $\cal M$ is defined as 
\begin{alignat*}{4}
	\tau, i & \models \text{true}\\
	\tau, i &\models p &&\text{iff }&& p \in \llangle s_i \rrangle \text{ where }\tau = s_0, s_1, s_2, \dots\\
	\tau,i &\models \varphi_1 \wedge \varphi_2 \quad &&\text{iff } &&\tau,i \models \varphi_1 \text{ and } \tau,i \models \varphi_2\\
	\tau,i &\models \neg \varphi \quad &&\text{iff } &&\tau,i \not\models \varphi\\
 	\tau,i &\models \varphi_1 \until \varphi_2 &&\text{iff } &&
    \text{for some finite $k \geq i$, $\tau, k \models \varphi_2$ and}  \\
    &&&&&\text{$\tau, j \models \varphi_1$ for all $j = i, \dots, k-1$}
\end{alignat*}
and is lifted to the entire transition system by requiring that every initialised trajectory satisfies $\varphi$:
$$ \mathcal{M} \models \varphi \text{ iff } \tau,0 \models \varphi \text{ for all infinite initialised trajectories }\tau\text{ of }{\cal M}.$$

We also introduce the derived operators \textit{"eventually"} $\lozenge$ and \textit{"globally"} $\square$. The formula $\lozenge \varphi := \text{true} \until \varphi$ states that $\varphi$ must be true in some state on the trajectory. The formula $\square \varphi := \neg(\lozenge \neg \varphi)$ requires that $\varphi$ holds true in all states of the trajectory. We do not include the \textit{"next"} operator $\nxt$ from full LTL since we are interested in stutter-insensitive bisimulations, which do not preserve a system's stepwise behavior as expressed by the \textit{next}-operator.
It is a well-known fact that divergence-sensitive stutter-insensitive bisimulations preserve specifications expressable in $\ltln$~\cite{DBLP:books/daglib/baierkatoen}. Divergence-sensitivity is crucial to properly treat stutter-trajectories, i.e., trajectories that forever stutter inside the same equivalence class~\cite{DBLP:journals/jacm/NicolaV95}. However, for deterministic transition systems, each state has only one outgoing trajectory that either eventually leaves its equivalence class or stutters indefinitely. Therefore, any stutter-insensitive bisimulation on a deterministic system must be divergence-sensitive, as stated in Lemma~\ref{lem:stuttdiv}.
\begin{lemma}
    Every stutter-insensitive bisimulation on any deterministic labelled transition system is divergence-sensitive.
    \label{lem:stuttdiv}
\end{lemma}
\begin{proof}
     Let $\cal M$ be a deterministic transition system and $\simeq$ be a stutter-insensitive bisimulation on $\cal M$. Let $s \simeq t$ and assume $s \simeq$-divergent but $t \text{ not} \simeq$-divergent. As $t \text{ not} \simeq$-divergent, there exists a finite trajectory 
 $\tau = t,t_1,\dots,t_n,t'$ with
    $t \simeq t_i, \forall i \leq n \text{ and } t \not \simeq t'$, for some $n \geq 0$. This implies that there exists a state $u \simeq s$ with $s \not \simeq t' \in T(u)$. However, since $\cal M$ is deterministic and $s$ is $\simeq$-divergent, the unique trajectory $\tau = s, s_1, \dots$ initalised in $s$ satisfies $s \simeq s_i, \forall i \geq 0$, which is a contradiction. \qed 
\end{proof}

\begin{theorem}
    \label{thm:divsens}
    Let $\cal M$ be a deterministic labelled transition system.
    If $\simeq$ is a stutter-insensitive bisimulation on ${\cal M}$, then ${\cal M} \models \varphi$ if and only if ${\cal M}/_\simeq \models \varphi$ for any $\ltl_{\setminus \nxt}$ formula $\varphi$.
\end{theorem}
\begin{proof}
 Any divergence-sensitive stutter-insensitive bisimulation $\simeq$ on any (possibly non-deterministic) transition system $\cal M$ implies, for every 
 $\ltln$ formula $\varphi$, 
 the model checking problems ${\cal M} \models \varphi$ and ${\cal M}/_\simeq \models \varphi$ have the same answer.
 Lemma~\ref{lem:stuttdiv} establishes that, since $\cal M$ is deterministic and $\simeq$ is stutter-insensitive on $\cal M$,
 then $\simeq$ is also divergence-sensisitve. Therefore, the statement follows.  
 \qed
\end{proof}

\begin{remark}
    Theorem~\ref{thm:divsens} in general does not hold for nondeterministic transition systems, as can be seen by the counterexample in Figure~\ref{fig:ctxnondet}.
\begin{figure}[h]
\centering

\scalebox{0.93}{
	\hspace*{4em}\begin{tikzpicture} [node distance = 2cm, on grid, auto]
 
 \begin{scope}[name prefix =G1, xshift= 0cm]
        \node (q0) [state, initial text = {}, fill = blue!30!white] {$s_0$};
        \node (q1) [state, initial,initial text = {},below = of q0, fill = blue!30!white] {$s_1$};
        \node (q2) [state, right = of q0,fill = red!30!white,yshift = -1cm] {$s_2$};
           \node [below of=q1, yshift = 1cm, xshift = 2.4cm] {\parbox{0.3\linewidth}{{\large{$\cal M$}}\label{subfig:a}}};
        \path [-stealth, thick]
            (q0) edge node {$ $}   (q2)
            (q1) edge node {$ $}  (q2)
            (q0) edge [loop above]  node {$ $}()
            (q2) edge [loop above]  node {$ $}();
        \path [dashed, thick]
            (q0) edge node {$ $} (q1);
 \end{scope}
 
 \begin{scope}[xshift = 5cm, name prefix =G2, yshift = -1cm]
        \node (q0) [state, initial, initial text = {}, fill = blue!30!white] {$R$};
        \node (q1) [state, right = of q0,,fill = red!30!white] {$Q$};
           \node [below of=q0, xshift = 2.4cm, yshift = 0cm] {\parbox{0.3\linewidth}{{\large{$\cal M/_{\simeq}$}}\label{subfig:a}}};
        \path [-stealth, thick]
            (q0) edge node {$ $}   (q1)
            (q0) edge [loop above]  node {$ $}()
            (q1) edge [loop above]  node {$ $}();
    \end{scope}
\end{tikzpicture}
}
\caption{A stutter-insensitive but not divergence-sensitive bisimulation $\simeq$ is indicated by the dashed lines. It holds that ${\cal M} \models \lozenge(\text{red})$ but ${\cal M/_{\simeq}} \not  \models \lozenge(\text{red})$.}
\label{fig:ctxnondet}
\vspace{-5mm}
\end{figure}
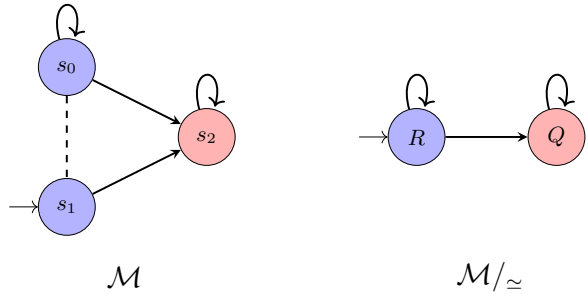

\end{remark}

\begin{remark}
It may seem counterintuitive to the reader to relate stutter-insensitive bisimulation and $\ltln$-equivalence, as it is usually associated with the more expressive $\ctln$~\cite{DBLP:books/daglib/baierkatoen}. However, recall that we focus on deterministic transition systems for which both logics coincide in expressivity.
\end{remark}

\section{Counterexample-guided Bisimulation Learning}

This section introduces our main contributions. We present our adaption of Namjoshi's well-founded bisimulations to deterministic transition systems~\cite{DBLP:conf/fsttcs/Namjoshi97} and describe a counterexample guided learning algorithm for simultaneous computation of a stutter-insensitive bisimulation and its corresponding quotient from a finite sampling of the state space. Well-founded bisimulation implements the stability conditions of stutter-insensitive bisimulation (see Def.~\ref{def:stuttBisim}) in the form of ranking functions that map states to a well-founded set, ensuring that a finite trajectory matches every transition of equivalent states. We present an adaption to deterministic systems that only requires the ranking functions to map single states of certain classes to a well-founded set and show that this characterizes stutter-insensitive bisimulation. Furthermore, we show that the problem of computing a stutter-insensitive bisimulation can be rephrased to finding a classifier on states and ranking functions for the corresponding classes.

\begin{theorem}
    Let $\cal M$ be a deterministic labelled transition system with state space $S$ and transition function $T$. Let $\simeq$ 
    be a label-preserving partition on $\cal M$. Suppose that for every region $R \in S/_\simeq$ there exists a function 
    $h_R \colon S \to \bbbn$ such that,
    for every $R \neq Q \in T/_\simeq(R)$, the following condition holds:
    \begin{equation}
        \forall s \in R \colon 
        T(s) \in Q \lor [T(s) \in R \land h_R(s) > h_R(T(s))].
        \label{eqn:rank}
    \end{equation}
    Then, $\simeq$ is a stutter-insensistive bisimulation on ${\cal M}$.
    \label{thm:rankstutt}
\end{theorem}

\begin{proof} 
Let $s \simeq t$ such that $s \not \simeq s' \in T(s)$. This implies that $\exists R \not = Q \in S/_{\simeq} \colon s,t \in R \text {  and } s' \in Q$, and $Q \in T/_{\simeq}(R)$. Assume there exists no finite trajectory $\tau = t,t_1,\dots,t_n,t', n \geq 0$ with
    $t \simeq t_i, \forall i \leq n \text{ and } s' \simeq t'$. As $\cal M$ is deterministic, we only need to distinguish the two cases:
    \begin{itemize}
        \item The infinite trajectory $\tau = t,t_1,t_2,\dots$ stutters infinitely in $R$, i.e., $t_i \in R, \forall i \geq 1$. This contradicts the ranking property~\ref{eqn:rank}.

        \item There exists a finite trajectory $\tau = t,t_1,\dots,t_n,t', n \geq 0$ with $t' \in V \not = Q$ and $t_i \in R, \forall i \leq n$. However $t_n \in R$ and $t' \in V \not = Q$ is a contradiction to the ranking property~\ref{eqn:rank} only allowing for an exit to $Q$. \qed 
    \label{thm:rank}
    \end{itemize}

The ranking functions $h_R$ in Theorem~\ref{thm:rankstutt} ensure that if class $R$ has an outgoing transition to $Q$, for all states in $R$ either (1) their successor is in $Q$ or (2) their successor is in $R$ and $h_R$ decreases when transitioning. As the value of $h_R$ is bounded from below and must strictly decrease along any trajectory, no trajectory can stutter in $R$ indefinitely and must eventually enter $Q$ (see Fig.~\ref{fig:thm2int}). 

\begin{figure}
    \centering
        \scalebox{1}{
        \begin{tikzpicture}

                \draw[] plot[smooth cycle] coordinates {(.7,-.3) (4.3,-.3)  (4.7,2.3)  (.5,2.3)};

                \draw[] plot[smooth cycle] coordinates {(5.3,-.3) (7.1,-.3)  (7.6,2.3)  (5.8,2.3)};

                \draw [] (0.2,2.7) node (s) {\Large$R$};
                \draw [] (5.7,2.7) node (s) {\Large$Q$};
                \begin{scope}[scale = 0.82, every node/.append style={transform shape}, yshift = -.5cm, xshift = -.2cm]
                        \draw [] (5,2.5) node (s) {\Large$s$};
                        \draw [] (8,2.57) node (sp) {\Large$s'$};
                        \draw [] (1.5,1.52) node (t) {\Large$t$};
                        \draw [] (2.8,1.48) node (t1) {\Large$t_1$};
                        \draw [] (5,1.48) node (tn) {\Large$t_n$};
                        \draw [] (8,1.57) node (tp) {\Large$t'$};
                
                        \draw [] (1.5,.8) node (ht) {\color{darkgray}$h_R(t)$};
                        \draw [] (2.1,.8) node () {\color{darkgray}$>$};
                        \draw [] (2.8,.8) node (ht1) {\color{darkgray}$h_R(t_1)$};
                        \draw [] (3.45,.8) node () {\color{darkgray}$>$};
                        \draw [] (3.9,.8) node () {\color{darkgray}$\dots$};
                        \draw [] (4.3,.8) node () {\color{darkgray}$>$};
                        \draw [] (5,.8) node (htn) {\color{darkgray}$h_R(t_n)$};
                
                        \draw [-stealth](5.3,2.6) [out = 30, in = 150]to (7.7,2.6);
                        \draw [-stealth](5.3,1.6) [out = 30, in = 150]to (7.7,1.6);
                        
                        \draw [-stealth](1.8,1.5) -- (2.4,1.5);
                        \draw [](3.1,1.5) -- (3.4,1.5);
                        \draw [-stealth](4.2,1.5) -- (4.6,1.5);
                        \draw [] (3.85,1.5) node (dots) {\Large$\dots$};
                \end{scope}

    \end{tikzpicture}
     }
 
        \caption{Intuitive representation of Theorem~\ref{thm:rankstutt}. Since a state $s \in R$ has a successor in $Q$, the value of $h_R$ must strictly decrease along any trajectory through $R$. Therefore, all states in $R$ eventually transition to $Q$ after possible stuttering.}
        \label{fig:thm2int}
        \vspace{-5mm}
\end{figure}
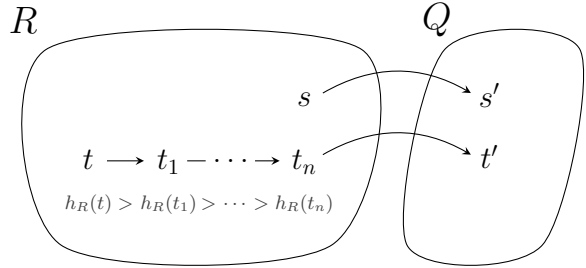

\begin{remark}
    For deterministic systems, strong bisimulations form a special case of stutter-insensitive bisimulations, not allowing for any stuttering. In our formulation, they only admit constant ranking functions $h_R$. Strong bisimulations preserve a system's stepwise behavior, hence, full LTL including the \textit{next}-operator~\cite{DBLP:books/daglib/baierkatoen}. However, they may induce much larger quotients, less suitable for verifying large systems with long stuttering intervals. Furthermore, there exist infinite state systems that do not admit a finite strong bisimulation, but do admit a finite stutter-insensitive bisimulation quotient.
\end{remark}

\end{proof}

We aim to phrase the problem of finding a suitable partition and ranking functions that satisfy the conditions in Theorem~\ref{thm:rankstutt} as a learning problem. For that, we introduce the notion of state classifiers.

\begin{definition}[State Classifier]
    A state classifier on a labelled transition system with state space $S$ is any function $f \colon S \to C$ that maps states to a finite set of classes $C$. It is label-preserving if $f(s) = f(t)$ implies 
    $\llangle s \rrangle = \llangle t \rrangle$.
\end{definition}

We can now state Theorem~\ref{thm:rankstutt} for a state classifier $f$ and give sufficient conditions for $f$ to induce a valid stutter-insensitive bisimulation. 

\begin{theorem}
    Let $\cal M$ be a deterministic labelled transition system with state space $S$ and transition function $T$. Suppose that there exists a label-preserving state classifier $f \colon S \to C$, a function  $g \colon C \to C$ and functions $h_c \colon S \to \mathbb{N}$ for each $c \in C$ such that, for every $c \neq d \in C$ and $s \in S$, the following two conditions hold:
\begin{align}
    &f(s) = c \land g(c) = d \implies f(T(s)) = d \lor [f(T(s)) = c \land h_c(s) > h_c(T(s))], \label{eqn:stateclass1}\\
    &f(s) = c \land f(T(s)) = d \implies g(c) = d.\label{eqn:stateclass2}
\end{align}
Then, $\simeq_f$ defined as $\simeq_f = \{(s,t) \mid f(s) = f(t)\}$ is a stutter-insensitive bisimulation on $\cal M$ and $T_{\simeq_f}(f^{-1}[c]) = \{f^{-1}[g(c)]\}$. 
\label{thm:stateclass}
\end{theorem}
\begin{proof}
We first show that $\simeq_f$ is a stutter-insensitive bisimulation on $\cal M$. Since $f$ is label-preserving, $\simeq_f$ is label-preserving by definition. The quotient space is the set of non-empty pre-images of the classes $C$ under $f$, i.e., $S/_{\simeq_f} = \{f^{-1}[c] \mid c \in C\} \setminus \emptyset$. By definition of $T/_{\simeq_f}$ (see Def.~\ref{def:quotient}) it holds that $f^{-1}[c] \neq f^{-1}[d] \in T/_{\simeq_f}(f^{-1}[c])$ implies that there exists an $s \in f^{-1}[c]$ with $T(s) \in f^{-1}[d]$. With Condition~\ref{eqn:stateclass2} this implies that $g(c) = d$. The claim follows by Condition~\ref{eqn:stateclass1} and Theorem~\ref{thm:rankstutt}. We now show that $T_{\simeq_f}(f^{-1}[c]) = \{f^{-1}[g(c)]\} $. Since $\simeq_f$ is a stutter-insensitive bisimulation Lemma~\ref{lem:detQuo} implies that $T/_{\simeq}(f^{-1}[c])$ can only be a singleton for any $c \in C$ . We distinguish the two cases:
\begin{itemize}
    \item $f^{-1}[c] \not = f^{-1}[d] \in T/_{\simeq_f}(f^{-1}[c])$, then $f^{-1}[c] \not = f^{-1}[d]$ implies that $c \not = d$. By Def.~\ref{def:quotient} there must exist a $s \in f^{-1}[c]$ with $T(s) \in f^{-1}[d]$, which by Condition~\ref{eqn:stateclass2} implies that $g(c) = d$. 
    \item $f^{-1}[c] \in T/_{\simeq_f}(f^{-1}[c])$, then some state in $s \in f^{-1}[c]$ must be $\simeq_f$-divergent by Def.~\ref{def:quotient}. The only possibility for $g$ to be a total function and not to violate Condition~\ref{eqn:stateclass1} is $g(c) = c$, as $g(c) = d \not = c$ would contradict the $\simeq_f$-divergency of $s$ due to Condition~\ref{eqn:stateclass1}.\qed 
\end{itemize}
\end{proof}
\begin{remark}
    Note that Theorem~\ref{thm:stateclass} requires $g$ to be well-defined, i.e., represent a deterministic transition function. However, this is not a restriction as per Lemma~\ref{lem:detQuo} any stutter-insensitive bisimulation on a deterministic transition system has a deterministic quotient. The fact that $g$ has to be total additionally requires it to correctly account for the self-loops of the divergent classes.
\end{remark}

In Theorem~\ref{thm:stateclass} function $g$ takes on the role of the deterministic transition function of the quotient induced by $f$. Thus, $f$ and $g$ together provide a complete description of a stutter-insensitive bisimulation quotient of the underlying transition system. In the following, we introduce our counterexample-guided learning approach for generating appropriate functions on a given transition system.

\subsection{Learner-Verifier Framework for Bisimulation Learning}

\begin{figure}[]
\vspace{-5pt}
\centering
\resizebox{1\textwidth}{!}{%
\begin{tikzpicture}

\begin{scope}[on background layer]

 \end{scope}

\draw [ fill={rgb,255:red,28; green,155; blue,181} , line width=0.2pt , fill opacity = 0.15, text opacity = 1, rounded corners] (8.25,11.03) rectangle  node (quotientsynth) {} (13.75,8.2);

\draw [ fill={rgb,255:red,28; green,155; blue,181} , line width=0.2pt , fill opacity = 0.3, text opacity = 1, rounded corners] (8.5,10.25) rectangle  node (learner) {\scriptsize Learner} (10.5,9);

\draw [ fill={rgb,255:red,28; green,155; blue,181} , line width=0.2pt , fill opacity = 0.3, text opacity = 1, rounded corners] (11.5,10.25) rectangle  node (certifier) {\scriptsize Verifier} (13.5,9);

\node[] at (7.65,10.81) {\scriptsize$f$};
\node[] at (7.65,8.4) {\scriptsize UNSAT};

\draw [ fill={rgb,255:red,237; green,81; blue,133}
, line width=0.2pt , fill opacity = 0.3, text opacity = 1, rounded corners, align = center] (4.7,10.25) rectangle  node (subquo) {\scriptsize Partition\\[-3pt]\scriptsize Template} (6.7,9);

\draw[-stealth] ([yshift=3.6em, xshift=0em]subquo.north) to[bend left = 0] ([yshift=.9em, xshift = 0em]subquo.north);
\node[] at (5.7,11.4) {\scriptsize$\llangle \cdot \rrangle$};

\draw[-stealth] ([yshift=4.05em, xshift=0em]learner.north) to[bend left = 0] ([yshift=1.35em, xshift = 0em]learner.north);

\node[] at (9.5,11.4) {\scriptsize$D, g, \{h_c\}_{c \in C}$};
\node[] at (12.7,11.4) {\scriptsize$S, T, g,\{h_c\}_{c \in C}$};

\draw[-stealth] ([yshift=4.05em, xshift=0em]certifier.north) to[bend left = 0] ([yshift=1.35em, xshift = 0em]certifier.north);

\draw[-stealth] ([yshift=.85em, xshift=1em]subquo.north) to[bend left = 21] ([yshift=1.4em, xshift = -1em]learner.north);
\draw[-stealth] ([yshift=-1.3em, xshift=-1em]learner.south) to[bend left = 21] ([yshift=-0.9em, xshift = 1em]subquo.south);

\draw[-stealth] ([yshift=1.3em, xshift=1em]learner.north) to[bend left = 27] ([yshift=1.4em, xshift = -1em]certifier.north);
\draw[-stealth] ([yshift=-1.3em, xshift=-1em]certifier.south) to[bend left = 27] ([yshift=-1.4em, xshift = 1em]learner.south);

\node[] at (11,10.8) {\scriptsize$\hat{\theta}, \hat{\gamma}, \hat{\eta}, f$};
\node[] at (11,8.45) {\scriptsize$\hat{s}_{cex}$};
\node[] at (11,7.9) {\scriptsize Bisimulation Learning};

\draw [ fill={rgb,255:red,81; green,237; blue,91} , line width=0.2pt , fill opacity = 0.4, text opacity = 1, rounded corners, align = center] (15.5,10.25) rectangle  node (veri) {\scriptsize Model\\[-3pt] \scriptsize Checking} (17.5,9);

\draw[-stealth] ([yshift=3.6em, xshift=0em]veri.north) to[bend left = 0] ([yshift=.9em, xshift = 0em]veri.north);

\draw[-stealth] ([yshift=0em, xshift=1.35em]certifier.east) to[bend left = 0] ([yshift=0em, xshift = -1em]veri.west);

\node[align = center] at (16.5,11.4) {\scriptsize $\varphi$};

\node[align = center] at (14.62,9.9) {\scriptsize $f(\hat{\theta}; \cdot), g(\hat{\gamma}; \cdot)$};

\draw[-stealth] ([yshift=-.8em, xshift=-.85em]veri.south) to[bend left = 0] ([yshift=-2.4em, xshift = -1.6em]veri.south);

\draw[-stealth] ([yshift=-.8em, xshift=.85em]veri.south) to[bend left = 0] ([yshift=-2.4em, xshift = 1.6em]veri.south);

\node[align = center] at (15.7,8.1) {\scriptsize \cmark \\[-4pt] \tiny Property\\[-5pt] \tiny Satisfied};
\node[align = center] at (17.25,8.08) {\scriptsize \xmark \\[-4pt] \tiny Counter-\\[-5pt] \tiny example};

\end{tikzpicture}
}
\caption{Architeture of our learner-verifier framework for bisimulation learning.}
\label{fig:cegis}
\end{figure}
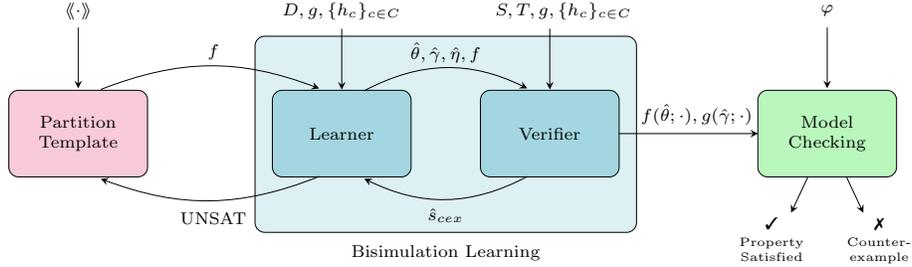

Our procedure involves two communicating components, the \textit{learner} and the \textit{verifier}, implementing a CEGIS loop. The learner proposes candidate functions that satisfy the stutter-insensitive bisimulation conditions over a finite set of sample states. The verifier checks if a counterexample state exists for which the functions proposed by the learner violate the conditions, which are then passed back to the learner to update the functions (see Figure~\ref{fig:cegis}). 

\subsubsection{Learner} We consider parametric function templates whose maps solely depend on the provided parameters. Therefore, the learner seeks suitable parameters for a label-preserving state classifier template $f \colon \Theta \times S  \to C$, a transition function template $g \colon \Gamma \times C \to C$ and ranking function templates $h_c \colon H \times S \to \mathbb{N}$ for each $c \in C$, i.e., attempts to solve:

\begin{equation}
    \exists \theta \in \Theta, \gamma \in \Gamma, \eta \in H \colon \bigwedge_{\hat{s} \in D} \Phi_1(\theta, \gamma, \eta; \hat{s}, T(\hat{s})) \wedge \Phi_2(\theta, \gamma, \eta; \hat{s}, T(\hat{s})),
    \label{eqn:phi}
\end{equation}
where $\Phi_1$ encodes Condition~\ref{eqn:stateclass1} of Theorem~\ref{thm:stateclass}:
\begin{multline}
    \Phi_1(\theta, \gamma, \eta; s, s') = \bigwedge_{c \not = d \in C} f(\theta;s) = c \land g(\gamma;c) = d \implies \\
    f(\theta;s') = d \lor [f(\theta; s') = c \land h_c(\eta;s) > h_c(\eta;s')],
\end{multline}
and $\Phi_2$ represents Condition~\ref{eqn:stateclass2}:
\begin{equation}
    \Phi_2(\theta, \gamma, \eta; s, s')= \bigwedge_{c \not = d \in C} f(\theta;s) = c \wedge f(\theta;s') = d \implies g(\gamma;c) = d,
\end{equation}
for a deterministic transition system $\cal M$ and finite set of sample states $D \subseteq S$. In our instantiation, we use an SMT-solver to seek a satisfying assignment for the parameters $\theta$, $\gamma$, and $\eta$ in the quantifier-free inner formula of \ref{eqn:phi}.

\subsubsection{Verifier}
The verifier checks the functions induced by the proposed candidate parameters $\hat{\theta}$, $\hat{\gamma}$ and $\hat{\eta}$ for generalisation to the entire state space, i.e., attempts to solve:

\begin{equation}
    \exists s \in S \colon \lnot \Phi_1(\hat{\theta}, \hat{\gamma}, \hat{\eta}; s, T(s)) \lor \lnot \Phi_2(\hat{\theta}, \hat{\gamma}, \hat{\eta}; s, T(s)).
    \label{eqn:cert}
\end{equation}

Similar to the learner, the verifier is an SMT-solver to which we hand the quantifier-free inner formula of~\ref{eqn:cert}. A found satisfying assignment for a counterexample state $s$ is returned to the learner. If the formula is unsatisfiable, the procedure terminates and has successfully synthesised a valid stutter-insensitive bisimulation and its corresponding quotient.

\subsection{Binary Decision Tree Partition Templates}
From here on, our focus is on transition systems with discrete state spaces $S \subseteq \mathbb{Z}^n$ defined over the integers. We present the parametric function templates used in our instantiation of the framework. For the state classifier templates, we employ binary decision trees with real-valued decision functions in the inner nodes. We construct binary decision trees preserving the system's labelling function and automatically enlarge them when the template is not expressive enough to fit the finite set of sample states (see Figure~\ref{fig:cegis}).
\begin{definition}[Binary Decision Tree Templates]
The set of binary decision tree templates $\mathbb{T}$ over a finite set of classes $C$ and parameters $\Theta$ consists of trees $t$, where $t$ is either 
\begin{itemize}
    \item a leaf node $\textsc{leaf}(c)$ with $c \in C$, or
    \item a decision node $\textsc{node}(p, t_1, t_2)$, where $t_1,t_2 \in \mathbb{T}$ are the left and right subtrees, and $p \colon \Theta \times S \to \mathbb{R}$ is a parametrised real-valued function of the states.     
\end{itemize}
A parametric tree template $t \in \mathbb{T}$ over classes $C$ and parameters $\Theta$ defines the parametric state classifier $f_t \colon \Theta \times S \to C$ given as  
\begin{align*}
    f_t(\theta, s) = \begin{cases}
                c &\text{if } t = \textsc{leaf}(c)\\
                f_{t_1}(\theta, s) &\text{if } t = \textsc{node}(p, t_1, t_2) \text{ and } p(\theta; s) \geq 0\\
                f_{t_2}(\theta, s) &\text{if } t = \textsc{node}(p, t_1, t_2) \text{ and } p(\theta; s) < 0.
            \end{cases}
\end{align*}
\end{definition}

\begin{figure}
    \centering
    \scalebox{0.87}{
        \begin{tikzpicture}
     \begin{scope}[name prefix =G1,xshift=3.5cm]

        \node[] (russell) 
            {\includegraphics[width=5cm]{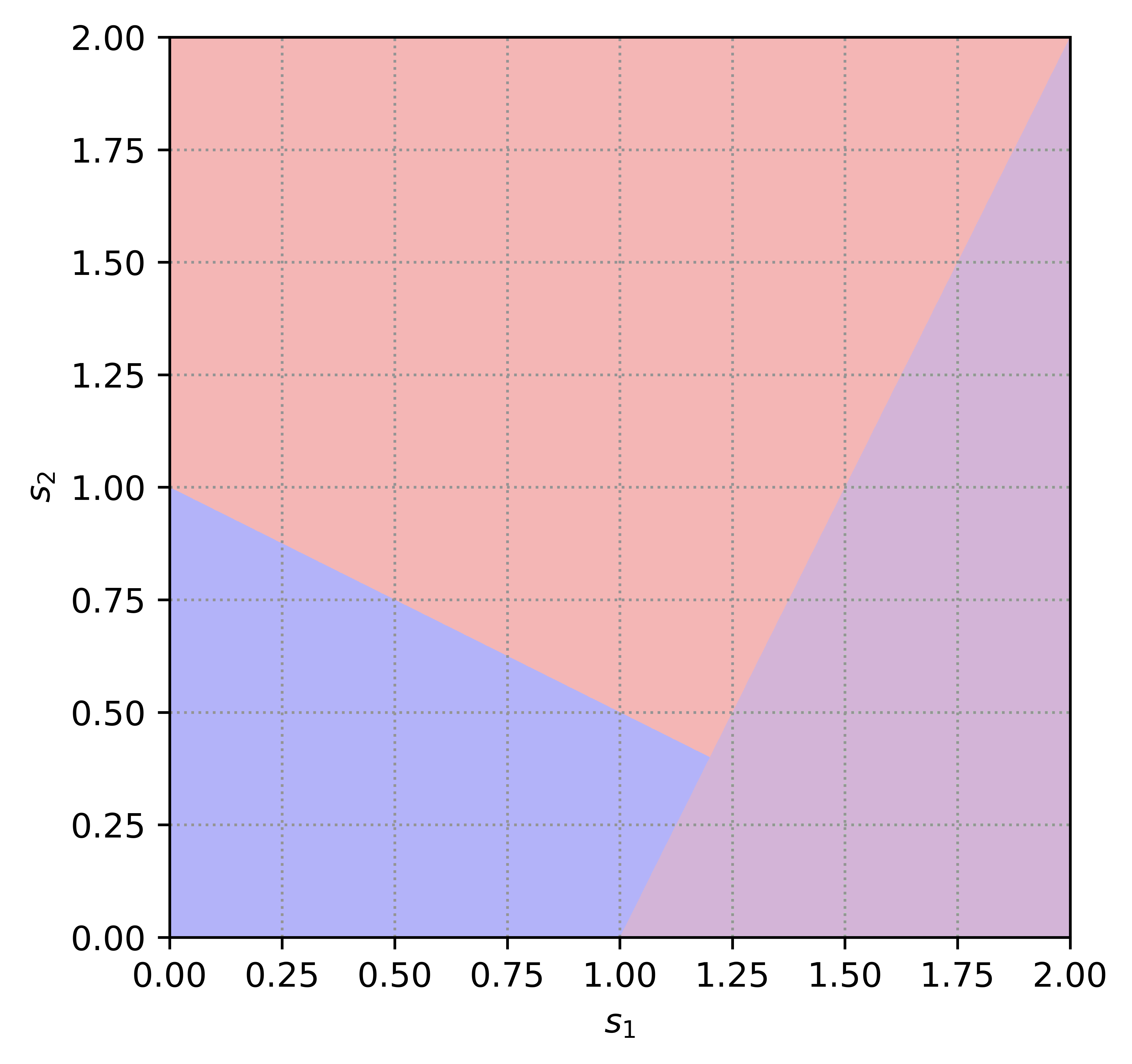}};
        \node[opacity=0.7] at (1.48, -1.05) (a) {\scriptsize$a$};
        \node[opacity=0.7] at (-1.0, -1.05) (b) {\scriptsize$b$};
        \node[opacity=0.7] at (-0.52, 1.45) (c) {\scriptsize$c$};
    \end{scope}
    
    \begin{scope}[xshift = -3cm,yshift=3cm, name prefix =G2, node distance = 1cm]
        \node (q0) [state, initial text = {}, fill = white, yshift = -1.5cm ] {\tiny$\theta_1s - \theta_2$};
        \node (q8) [state, below left = of q0, fill = violet!30!white] {$a$};
        \node (q4) [state, below right = of q0,fill = white] {\tiny$\theta_3s - \theta_4$};
        \node (q2) [state, below left = of q4,fill = blue!30!white] {$b$};
        \node (q3) [state, below right = of q4,fill = red!30!white] {$c$};

            \path [-stealth, thick]
            (G2q0) edge [] node [] {} (G2q4)
            (G2q0) edge [] node [] {} (G2q8)
            (G2q4) edge [] node [] {} (G2q2)
            (G2q4) edge [] node [] {} (G2q3);

    \end{scope}
    \end{tikzpicture}
    }
    \caption{Binary decision tree with parameters $\theta_1 = \begin{bmatrix}
	-2 & 1
\end{bmatrix}, \theta_2 = -2, \theta_3 = \begin{bmatrix}
	\frac{1}{2} & 1
\end{bmatrix}$, and $\theta_4 = 1$ for the parametrised functions, and its corresponding state classifier.} 
    \label{fig:bdt}
\end{figure}
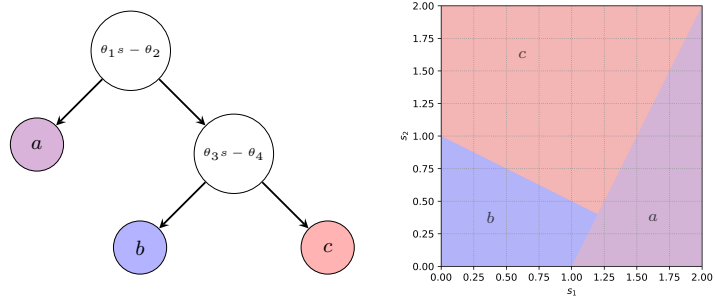

Binary decision trees appeal as state classifier templates as they are interpretable, expressive, and simple to translate into quantifier-free expressions to instantiate the formulas of the learner and the verifier, see Figure~\ref{fig:bdt}. The parametric transition function template $g \colon \Gamma \times C \to C$ is simply a vector or list over classes $C$, indexed by $C$, and for the parametric ranking function templates $h_c \colon H \times S \to \mathbb{N}$ we consider linear functions of the form $h_c(\eta, s) = \eta_1 \cdot s + \eta_2$.

An important requirement for our procedure is that the synthesised state classifier is label-preserving. We guarantee this by constructing label-preserving templates which have this property by design for any parameter instantiation. We assume that any atomic proposition $a \in AP$ is associated with a real-valued function $p_a \colon S \to \mathbb{R}$, such that 

\begin{equation}
    \llangle s \rrangle = \{ a \in AP \mid p_a(s) \geq 0 \}. 
\end{equation}

We 
construct label-preserving templates by encoding the functions corresponding to the atomic propositions into the top nodes of the binary decision tree, i.e., fixing the functions for the top nodes to represent the observation partition (see for example~\cite{DBLP:journals/automatica/Pappas03} for a canonical construction). This resembles the prerequisite of classical partition-refinement algorithms, which are initialised from label-preserving partitions. Fixing the labelling with the top nodes ensures that any instantiated state classifier is label-preserving, and the subsequent nodes further refine the label-preserving partition. Figure~\ref{fig:polylab} shows the top nodes with fixed functions for a label-preserving binary tree template for the Euclidean algorithm from Figures~\ref{fig:euclid} and~\ref{fig:exProgress}.

\begin{figure}
    \centering
    \scalebox{0.8}{
    \begin{tikzpicture}    
    \begin{scope}[xshift = -3cm,yshift=3cm, name prefix =G2, node distance = 1cm]
        \node (q0) [state, initial text = {}, fill = white, yshift = -1.5cm ] {\small$x - y$};
        \node (q8) [below right = of q0] {\large$\vdots$};
        \node (q4) [state, below left = of q0,fill = white] {\small$y - x$};
        \node (q2) [state, below left = of q4,fill = lcyan] {$t$};
        \node (q3) [below right = of q4] {\large$\vdots$};

            \path [-stealth, thick]
            (G2q0) edge [] node [] {} (G2q4)
            (G2q0) edge [] node [] {} (G2q8)
            (G2q4) edge [] node [] {} (G2q2)
            (G2q4) edge [] node [] {} (G2q3);

    \end{scope}
    \end{tikzpicture}
    }
    \caption{Label-preserving binary decision tree template for the Euclidian algorithm from Figure~\ref{fig:euclid}. The functions are assigned with respect to the loop condition $x \not = y$. The \textit{terminated} proposition, i.e., class $t$, is assigned to the states satisfying $x \leq y$ and $y \leq x$. Any state satisfying either of $x > y$ or $y > x$ is labelled with \textit{not-terminated}.}
    \label{fig:polylab}
    
\end{figure}
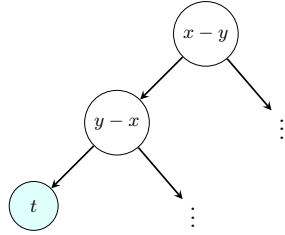

 When the binary decision tree used is too small to fit the minimum number of regions in a quotient or if it requires more decision boundaries, the learner will return UNSAT as it cannot fit a partition with the given template on the finite set of samples. In such cases, our procedure automatically increases the size of the employed BDT template and resumes bisimulation learning with the more expressive template (see Figure~\ref{fig:cegis}). Our template construction is entirely automatic and requires no user input other than the labeling function. Bisimulation learning starts from a small, automatically generated BDT template encoding the labeling function and successively enlarges the partition template as required. We enlarge the partition template by adding an additional layer to the BDT, doubling the number of available partitions.

\section{Experimental Evaluation}

We implemented our approach in a software prototype and evaluated bisimulation learning on a range of benchmark systems representing two common classes of problems: verification of reactive systems and software model checking. We compare our procedure to established state-of-the-art tools: the nuXmv model checker~\cite{DBLP:conf/cav/CavadaCDGMMMRT14,DBLP:journals/iandc/BurchCMDH92,DBLP:conf/vmcai/Bradley11} for the reactive system problems and the Ultimate~\cite{DBLP:conf/tacas/HeizmannBDFHKNSSP23} and CPAChecker~\cite{DBLP:conf/cav/BeyerK11} tools for software model checking benchmarks. All benchmarks, our implementation, and the used templates are publicly available. We employ the Z3 SMT-solver~\cite{DBLP:conf/tacas/MouraB08} in both learner and verifier, and the nuXmv model checker to verify the properties of interest on the obtained abstractions.

\subsection{Discrete-Time Clock Synchronization}
\subsubsection{Setup}
For reactive systems, we consider two distributed synchronization protocols for potentially drifted discrete clocks of distributed agents. First, the TTEthernet protocol, where all agents send their current clock value to a central synchronization master. This synchronization master computes the median clock value and sends it back to the agents, which use the received value to update their internal clocks~\cite{DBLP:conf/cpsweek/BogomolovHS16}. Second, we consider an interactive convergence algorithm where the agents directly exchange clock values and compute the average to update their internal clocks while excluding received values that differ more than a given threshold from their own~\cite{DBLP:conf/podc/LamportM84}. We check the systems for two kinds of properties: a safety invariant, which specifies that all clock valuations remain within a predefined maximum distance (G(safe)); and whether all clocks infinitely often synchronize on the same valuation (GF(sync)). Note that while the baseline procedures verify the systems regarding the given specification, our abstraction procedure is agnostic to the specification, i.e., the obtained abstraction can be used to verify arbitrary $\ltln$ formulas over the atomic propositions. 

To render verification with BDDs feasible, we leverage that all clock valuations remain within an interval that depends on the time discretization, as they are either continuously reset or enter a dead-lock state when violating the safety requirement. We explicitly pass this invariant to the BDD toolchain in the form of finite variable domains to allow for the construction of BDDs, whereas IC3 and our abstraction approach operate over unbounded integer variables, which would not be possible for BDDs.
For both benchmarks, we consider a safe variant, where the agents use the received values to update their internal clock correctly, and an unsafe version, where they stick with their internal values and drift further. In all instances, we assess multiple instances of time discretization, i.e., the sampling frequency (number of discrete time steps) for a unit second.  
\subsubsection{Results}
Table~\ref{tab:clock} presents the runtime results.
Our approach depends on generating candidate parameters and counterexamples through an SMT solver. These can vary across runs of the procedure, even under identical initial conditions (i.e., provided initial samples).  
Since this can impact the convergence speed and overall runtime of the algorithm, we conduct each experiment 10 times and report the average runtimes and standard deviations. We only report a single outcome for our approach, as we check both properties of interest on the same abstraction and the differences in verification time are negligible on the obtained small abstractions. 

\newcolumntype{N}{>{\centering\arraybackslash}m{2.3cm}}
\newcolumntype{M}{>{\centering\arraybackslash}m{1.8cm}}

\begin{table}[h]
\centering
\caption{Results for reactive clock-synchronization benchmarks. All times are measured in seconds with ``oot'' denoting a timeout at 500 [sec]. The benchmark names include the used parameters, e.g., ``tte-sf-1k'' describes a safe TTEthernet instance with a time discretization of 1000 steps per second.}

\scalebox{0.71}{
\begin{tabular}{|c|c||N|N||N|N||N|}
\hline
Benchmark & No.\ States  & \multicolumn{2}{c||}{nuXmv (IC3)} &
\multicolumn{2}{c||}{nuXmv (BDDs)} & Bisimulation\\
&&  G(safe) &GF(sync)&G(safe)&GF(sync) & Learning
\\\hline
tte-sf-10  & 250   &  0.1 & 0.7 & 0.1 & 0.1& 0.3\tiny$\pm 0.4$ \\
tte-sf-100 & 2500   & 13.3 & 423 & 0.3  & 0.3& 0.7\tiny$\pm 0.6$  \\
tte-sf-1k & \num{2.5e6}   & oot & oot & 1.8  & 31& 1.2\tiny$\pm 0.4$ \\
tte-sf-2k & \num{1e7}   & oot & oot & 6.4  & 162& 1.5\tiny$\pm 0.1$  \\
tte-sf-5k & \num{6.25e7}   & oot & oot & 34  & 417 & 1.6\tiny$\pm 0.4$ \\
tte-sf-10k & \num{2.5e8}   & oot & oot & 193  & oot& 1.6\tiny$\pm 0.2$ \\
\hline
tte-usf-10 & 250  & 0.1 & 0.5 & 0.1 & 0.1 & 0.2\tiny$\pm 0.1$  \\
tte-usf-100 & 2500  & 15.2 & 9.2 & 0.1 & 0.2& 0.2\tiny$\pm 0.1$ \\
tte-usf-1k & \num{2.5e6}   & 421 & 405& 1.5& 10& 0.3\tiny$\pm 0.1$ \\
tte-usf-2k & \num{1e7}   & oot & oot&  7.1  & 41 & 0.4\tiny$\pm 0.2$ \\
tte-usf-5k & \num{6.25e7}    & oot & oot & 32 &  242& 0.5\tiny$\pm 0.5$ \\
tte-usf-10k & \num{2.5e8}  & oot &oot & 130 & oot& 0.6\tiny$\pm 0.5$  \\
\hline
con-sf-10  & 250   &  0.6 & 0.5 & 0.1  & 0.1& 0.4\tiny$\pm 0.2$\\
con-sf-100 & 2500  & 22 & oot & 0.2  & 0.3&0.4\tiny$\pm 0.3$ \\
con-sf-1k & \num{2.5e6}   & oot & oot & 4.6  & 45& 0.7\tiny$\pm 0.4$\\
con-sf-2k & \num{1e7}   & oot & oot & 17.6 & 210& 0.7\tiny$\pm 0.2$ \\
con-sf-5k & \num{6.25e7} & oot & oot & 95 & oot & 0.8\tiny$\pm 0.5$\\
con-sf-10k & \num{2.5e8}   & oot & oot & oot & oot& 1.4\tiny$\pm 1.4$\\
\hline
con-usf-10  & 250  &  0.2 & 0.3 & 0.1  & 0.1 & 0.2\tiny$\pm 0.2$\\
con-usf-100 & 2500  & 31 & 33 & 0.3 & 0.2& 0.2\tiny$\pm 0.1$\\
con-usf-1k & \num{2.5e6}   & oot & oot & 2.8 & 24& 0.3\tiny$\pm 0.2$\\
con-usf-2k & \num{1e7}   & oot & oot & 8.2 & 154 & 0.4\tiny$\pm 0.3$\\
con-usf-5k & \num{6.25e7} & oot & oot& 36  & oot  & 0.7\tiny$\pm 0.4$\\
con-usf-10k & \num{2.5e8} & oot & oot& 156  & oot& 0.8\tiny$\pm 0.3$  \\
\hline
\end{tabular}

}

\label{tab:clock}
\vspace{-8mm}
\end{table}

\subsubsection{Discussion}
The results show that the learned bisimulations effectively and efficiently verify the specifications. Especially with decreased time discretization and, therefore, increased size of the state space, our approach clearly shows an advantageous performance. While a larger reachable state space renders verification harder for all considered approaches, a decreased time discretization is especially difficult for the IC3 toolchain based on bounded model checking, as it increases the completeness threshold and the depth of counterexamples. Since bisimulation learning generalizes from a finite set of samples, it is less susceptible to larger state spaces if the corresponding abstractions remain small. Generally, there is a trade-off between the number of provided initial samples and the number of CEGIS iterations needed to refine the initial partition. Although it may require more time to fit an initial partition on a more extensive set of uniform initial samples, it can reduce the counterexamples needed to obtain a valid stutter-insensitive bisimulation. As our instantiation leverages potentially expensive SMT solving in both the learner and the verifier, which scales in the number of considered samples, we aim at being sample-efficient: therefore, we decided to consider a fixed, small amount of uniform initial samples for all benchmarks of different sizes and leverage the generation of informative counterexamples in potentially more, but faster CEGIS cycles.

\subsection{Conditional Termination}
\subsubsection{Setup}
For software model checking, we consider a range of benchmarks from program termination analysis, including a selection of programs sourced from the termination category of the SV-COMP competition for software verification~\cite{DBLP:conf/tacas/Beyer23}. As is the case for the Euclidean algorithm in Figure~\ref{fig:euclid}, these programs on unbounded integer variables may terminate for some inputs and enter a non-terminating loop for others. The two baseline tools determine whether a program terminates for all possible inputs. Our procedure instead goes a step further by providing an exact partition of the variable valuations, separating the inputs for which the algorithm eventually terminates from those for which it does not. As a distinguishing feature of the baseline benchmarks, we split each program into two versions: one that only allows for inputs for which the program terminates (denoted as ``term'') and another that includes potentially non-terminating inputs (denoted as ``$\neg$term``). 

\subsubsection{Results}
Table~\ref{tab:term} presents the runtime results. Note that we only report the analysis time for the baselines, without additional time spent on parsing or preprocessing the programs. 

\begin{table}[h]
\centering
\caption{Results for software termination benchmarks. All times are measured in seconds, with ``oot'' denoting a timeout at 500 [sec]. A non-conclusive analysis outcome is denoted by ``n/c'' and ``-'' indicates that there is no such special case of the benchmark.} 
\scalebox{0.72}{
\begin{tabular}{|c||M|M||M|M||M|M||M|}
\hline
Benchmark  & \multicolumn{2}{c||}{nuXmv (IC3)} &
\multicolumn{2}{c||}{CPAChecker}&
\multicolumn{2}{c||}{Ultimate} & Bisimulation\\
&  \small term & \small $\neg$term&  \small term & \small $\neg$term&\small term & \small $\neg$term& Learning
\\\hline
term-loop-1    &  oot & oot & 0.6 & 1.8 & 1.5 & 2.9  & 0.2\tiny$\pm 0.1$\\
term-loop-2    &  oot & oot & 0.6 &  0.3 & 0.7 & 0.2  & 0.4\tiny$\pm 0.3$\\
audio-compr & 3.1 & $<0.1$& n/c & n/c  & 0.9 & 0.6 &  0.3\tiny$\pm 0.3$ \\
euclid  & oot & oot &n/c & 0.3& 1.6 & 0.4 &  0.6\tiny$\pm 0.2$\\
greater  & oot & $<0.1$  & 0.6 & 0.3 & 1.1 & 0.3 &  0.4\tiny$\pm 0.2$\\
smaller  &  oot & $< 0.1$ & 0.6 & 0.3 & 1.6 & 0.3 &  0.2\tiny$\pm 0.1$\\
conic  & oot &  $< 0.1$ & 0.7 & 0.4 & n/c & 0.4 & 4.2\tiny$\pm 7.3$ \\
disjunction & oot & - & 34 & - & 1.7 & -  & 0.3\tiny$\pm 0.3$ \\
parallel & n/c & - & 0.9 & - & 9.1 & - & 0.3\tiny$\pm 0.3$ \\
quadratic & $0.2$ & - & n/c & - & n/c & - & 0.3\tiny$\pm 0.1$  \\
cubic & 0.2 & $< 0.1$  & n/c & n/c & n/c & 0.2  & 0.4\tiny$\pm 0.2$ \\
nlr-cond & $< 0.1$&  $< 0.1$ & n/c & n/c & n/c & 0.2  & 0.2\tiny$\pm 0.2$ \\
\hline
\end{tabular}

}

\label{tab:term}
\vspace{-4mm}
\end{table}

\subsubsection{Discussion}

The results show that bisimulation learning, while computing more informative results and solving the more complex problem of conditional termination~\cite{DBLP:conf/tacas/BozgaIK12,DBLP:conf/cav/CookGLRS08}, operates in runtimes comparable to the state-of-the-art tools for the considered benchmarks. Especially for programs that involve disjunctions over variable valuations (cf.\ the \textit{disjunction} and \textit{parallel} benchmarks), our procedure is able to prove termination more efficiently. Additionally, our approach can handle non-linear operations if the employed templates are sufficiently expressive for the corresponding partition and ranking functions. 
As a further surplus, it yields interpretable binary decision trees representing the derived stutter-insensitive bisimulation. These trees are valuable for system diagnostics and fault analysis, providing further insight beyond single counterexamples. Once again, this experimental evaluation shows that, while not being complete in theory, our algorithm terminates in all of the considered experiments. 

\subsubsection{Limitations}

Bisimulation learning addresses a generally undecidable problem: finding finite bisimulations for systems with potentially infinite state spaces~\cite{DBLP:conf/concur/Moller96}. While our procedure is guaranteed to terminate on finite state systems, it must be inherently incomplete in general. Our experimental evaluation demonstrates that we can effectively and efficiently find finite bisimulations for infinite-state systems. However, there exist systems for which bisimulation learning can never successfully terminate. We give an example for such a system: Consider the infinite state space of natural numbers $S = \{0, 1, \dots \}$, where each state transitions by subtracting one, and zero loops on itself, i.e., $T = \{0 \mapsto 0\} \cup \{n \mapsto n - 1, n > 0\}$. The labelling function distinguishes \textit{zero}, \textit{even}, and \textit{odd} numbers (see Figure~\ref{fig:ctx}).

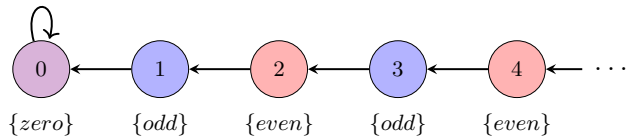
\begin{figure}
\centering

\scalebox{0.93}{
	\begin{tikzpicture} [node distance = 1.7cm, on grid, auto]
 
 \begin{scope}[name prefix =G1, xshift= 0cm]
        \node (q0) [state, initial text = {}, fill = violet!30!white, label={[label distance=.1cm]below:$\{zero\}$}] {$0$};
        \node (q1) [state,initial text = {},right = of q0, fill = blue!30!white, label={[label distance=.1cm]below:$\{odd\}$}] {$1$};
        \node (q2) [state, right = of q1,fill = red!30!white, label={[label distance=.1cm]below:$\{even\}$}] {$2$};
        \node (q3) [state, right = of q2,fill = blue!30!white, label={[label distance=.1cm]below:$\{odd\}$}] {$3$};
        \node (q4) [state, right = of q3,fill = red!30!white, label={[label distance=.1cm]below:$\{even\}$}] {$4$};

        \path [-stealth, thick]
            (q1) edge node {$ $}  (q0)
            (q2) edge node {$ $}  (q1)
            (q3) edge node {$ $}  (q2)
            (q4) edge node {$ $}  (q3)
            (q0) edge [loop above]  node {$ $}();

        \draw[-stealth, thick] ([yshift=0em, xshift=1.6em]q4.east) to[bend left = 0] ([yshift=0em, xshift = 0em]q4.east);

        \draw [] (8.18,0.01) node (dots) {\textbf{\large$\dots$}};
        
 \end{scope}
 
\end{tikzpicture}
}
\caption{A system for which bisimulation learning can never terminate, as no finite stutter-insensitive bisimulation exists.}
\label{fig:ctx}
\vspace{-5mm}
\end{figure}

Any infinite trajectory starting in some state $i$ will eventually enter state zero. However, depending on the starting state, it will traverse a different sequence of even and odd states. Hence, we can construct $\ltln$ formulas that distinguish each state from \textit{smaller} states. For instance, the formula $\lozenge({even} \land \lozenge(\textit{odd} \land \lozenge(\textit{zero})))$ can only be satisfied by states larger than one. Per Theorem~\ref{thm:divsens}, since $\ltln$ can distinguish any state from smaller states, every state must be its own equivalence class with respect to stutter-insensitive bisimulation. When applying bisimulation learning to the described system, the CEGIS loop can never terminate with a finite quotient. Our procedure will keep enlarging the partition template used to fit the growing set of counterexamples, but will never be able to generalize to the entire state space. We note that this is a limitation intrinsic to bisimulations, i.e., no bisimulation algorithm could successfully terminate when applied to the stated system.

\section{Conclusion} 

We have presented the first data-driven method to compute bisimulations. We have demonstrated that our method effectively computes finite abstractions for  model checking and diagnostics. We instantiated our method to stutter-insensitive bisimulations, showcased its efficacy on $\ltln$ model checking of discrete-time synchronization protocols as well as on conditional termination analysis benchmarks from the SV-COMP. On these benchmarks, our method yielded faster results than alternative model checking algorithms based on BDDs and IC3 (nuXmv), and state-of-the-art software model checking procedures (Ultimate and CPAChecker). Our benchmark sets are systems with long completeness thresholds and deep counterexamples, for which stutter-insensitive bisimulations provide succinct abstract quotients. 

Our technique builds upon an existing proof rule for well-founded bisimulations. This allows us to characterise stutter-insensitive bisimulations as classifiers from infinite concrete states to finite abstract states, with an attached ranking function on each abstract state that strictly decreases as the concrete system stutters. This has enabled implementing a learner-verifier framework to compute bisimulations for deterministic systems with discrete state space. Our approach readily extends to \textit{strong} bisimulations for deterministic systems, even though in practice these produce too large abstractions for effective model checking. Stutter-insensitive bisimulations instead are coarser and, therefore, generate smaller quotients. Not only this enables an effective verification of $\ltln$ properties, but also provides succinct and interpretable abstractions. 

Our result is the basis for several extensions. First, we envision extensions towards stutter-insensitive bisimulations for non-deterministic systems, which are harder because they require more general conditions on learner and verifier. Second, we target  extensions towards continuous-state systems, which are harder because they offer much less flexibility in terms of numerical representation~\cite{DBLP:conf/hybrid/Girard07,DBLP:journals/ejcon/GirardP11,DBLP:journals/tac/ZamaniEMAL14}. Lastly, we envision extensions towards using neural architectures for further flexibility and scalability in state classifier representation~\cite{DBLP:conf/nips/AbateEG22}.

\bibliographystyle{splncs04}
\bibliography{bibliography}

\end{document}